%% file: randomness_security.tex
\documentclass[10pt,conference]{IEEEtran}
\usepackage{color,graphicx,amsmath,amssymb,stmaryrd,dsfont}
\usepackage{rotating,array,url}

\input{rv_defs.tex}
\input{fct_defs.tex}
\input{th_defs.tex}

\renewcommand{\leq}{\leqslant}
\renewcommand{\geq}{\geqslant}
\graphicspath{{graphics/}}

\begin{document}

\title{On Secure Communication with Constrained Randomization}

\author{\IEEEauthorblockN{Matthieu R. Bloch}
\IEEEauthorblockA{School of Electrical and Computer Engineering\\
Georgia Institute of Technology\\
Atlanta, Georgia 30332--0250\\
Email: matthieu.bloch@ece.gatech.edu}
\and
\IEEEauthorblockN{J\"org Kliewer}
\IEEEauthorblockA{School of Electrical and Computer Engineering\\
New Mexico State University\\
Las Cruces, New Mexico 88003-8001\\
Email: jkliewer@nmsu.edu}}

\maketitle

\begin{abstract}
  In this paper, we investigate how constraints on the randomization in the encoding process affect the secrecy rates achievable over wiretap channels. In particular, we characterize the secrecy capacity with a rate-limited local source of randomness and a less capable eavesdropper's channel, which shows that limited rate incurs a secrecy rate penalty but does not preclude secrecy. We also discuss a more practical aspect of rate-limited randomization in the context of cooperative jamming. Finally, we show that secure communication is possible with a non-uniform source for randomness; this suggests the possibility of designing robust coding schemes.
\end{abstract}

\section{Introduction}
\label{sec:introduction}

The wiretap channel model~\cite{Wyner1975,Csiszar1978} has attracted much attention in recent years because of its potential to strengthen the security of communication systems~\cite{InfoTheoreticSec,Bloch2011}. Although this model provides a convenient abstraction to design codes for secure communication (see~\cite{Mahdavifar2011} and reference therein), it relies on two implicit simplifying assumptions. First, the model assumes that the transmitter knows the statistics of the channel. Second, the model assumes that the transmitter has access to an arbitrary local source of randomness, whose statistics can be optimized as part of the code design. In practice, however, these assumptions are unlikely to be perfectly guaranteed. For instance, an eavesdropper has little incentive to help characterize the channel statistics and, realistically, the legitimate parties may only have approximate knowledge of the true statistics. Similarly, the statistics of the local source of randomness may be imperfectly known, or the source may only provide a limited rate of randomness.

Secure communications with imperfect channel knowledge have already been the subject of previous investigations. For instance, several works have studied \emph{compound} wiretap channels (see~\cite{Liang2009} and references therein), in which the transmitter only knows that its channel belongs to a set of possible channels. Secure communication is often possible but the best channel to the eavesdropper usually limits secrecy rates. Other works have investigated the secrecy capacity of state-dependent channels under different assumptions regarding state information (see~\cite{InfoTheoreticSec,Bloch2011} and references therein). In another approach,~\cite{Muramatsu2009} has shown the existence of \emph{universal} wiretap codes, which guarantee secrecy and reliability as soon as the channel capacity of the eavesdropper's channel is low enough.

In contrast to the problem of channel knowledge, little attention has been devoted to the problem of imperfect local sources of randomness. In particular, the questions of how much randomness is required to guarantee secrecy and how sensitive are secure communication codes to imperfections in randomness are still largely open.

In this paper, we provide partial answers to these questions. Our main contributions are 1)~the characterization of secrecy capacity with a rate-limited source of randomness and a less capable eavesdropper's channel, 2) practical considerations on the effect of limited randomness for cooperative jamming, and~3)~the derivation of a sufficient condition for secure communication with a non-uniform randomization. 

The remainder of the paper is organized as follows. Section~\ref{sec:problem-statement} introduces the wiretap channel model used to analyze the effect of constrained randomization and presents our results on the secrecy-capacity of wiretap channels with a rate-limited local source of randomness. Section~\ref{sec:discussion} discusses rate-limited randomness in the context of cooperative jamming. Finally, Section~\ref{sec:secure-comm-with} discusses the possibility of secure communication with a non-uniform local source of randomness that cannot be processed. 
 
\section{Rate-Limited Randomness: Theoretical Considerations}
\label{sec:problem-statement}

Unless otherwise specified, we consider a discrete wiretap channel $\left(\calX,W_{\rvY\rvZ|\rvX},\calY\times\calZ\right)$, characterized by a finite input alphabet $\calX$, two finite output alphabets $\calY$ and $\calZ$, and transition probabilities $p_{\rvY\rvZ|\rvX}$. As illustrated in Figure~\ref{fig:wireta_channel_model}, we assume that the transmitter (Alice) wishes to transmit a secret message to the receiver observing $\rvY^n$ (Bob), in the presence of an eavesdropper observing $\rvZ^n$ (Eve). The channel $\left(\calX,W_{\rvY|\rvX},\calY\right)$ is called the main channel while the channel $\left(\calX,W_{\rvZ|\rvX},\calZ\right)$ is called the eavesdropper's channel. We assume the eavesdropper's channel is less capable, that is for any input $\rvX$ we have $\avgI{\rvX;\rvZ}\leq \avgI{\rvX;\rvY}$. 
\begin{figure}[t]
  \centering
  \scalebox{0.5}{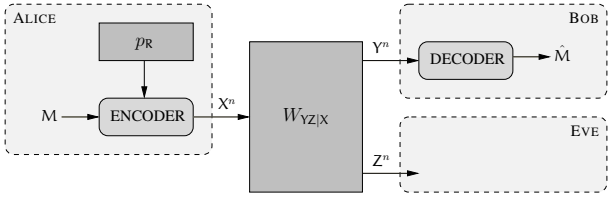}
  \caption{Communication over a randomness-limited wiretap channel.}
  \label{fig:wireta_channel_model}
\end{figure}  
The encoding process may be stochastic, but the only source of randomness is a discrete memoryless\footnote{The assumption of a memoryless source is a matter of convenience, and the proofs in the appendices generalize easily to arbitrary sources.} source $(\calR,p_R)$ with known alphabet $\calR$ and known statistics $p_\rvR$. This model captures a situation in which the transmitter does not have access to a infinite pool of random numbers, and those must be generated on-the-fly during encoding from a source of randomness (thermal noise, photon counting). In addition, it forces us to specify explicitly how to use the randomness provided by the source in the encoding process.\medskip{}

\begin{definition}
  A $(2^{nR},n)$ wiretap code $\calC_n$ for the discrete wiretap channel $\left(\calX,p_{\rvY\rvZ|\rvX},\calY\times\calZ\right)$ with local source of randomness $(\calR,p_R)$ consists of the following.
  \begin{itemize}
  \item a message alphabet $\calM = \intseq{1}{2^{nR}}$;
  \item an encoding function $e:\calM\times\calR^n\rightarrow\calX^n$;
  \item a decoding function $f:\calY^n\rightarrow \calM\cup\{?\}$.
  \end{itemize}
  The performance of $\calC_n$ is measured in terms of the average probability of error $P_e(\calC_n)\eqdef \P{\rvM\neq \hat\rvM|\calC_n}$ and of the secrecy leakage $L(\calC_n)\eqdef \avgI{\rvM;\rvZ^n|\calC_n}$
\end{definition}\medskip{}

\begin{definition}
  A rate $R$ is achievable if there exists a sequence of $(2^{nR},n)$ wiretap codes $\{\calC_n\}_{n\geq 1}$ such that
  \begin{align*}
    \lim_{n\rightarrow\infty}P_e(\calC_n)=0\quad\text{and}\quad\lim_{n\rightarrow\infty}L(\calC_n) = 0.
  \end{align*}
  The (strong) secrecy capacity with rate-limited randomness $C_s$ is defined as the supremum of all achievable rates.
\end{definition}

\begin{remark}
  The definition of a wiretap code above implicitly allows the encoder to process the observations obtained from the local source of randomness. In particular, the encoder can remove a possible bias in the randomness. What happens when the encoder does not perfectly process the local source is discussed in Section~\ref{sec:secure-comm-with}.
\end{remark}

\begin{proposition}
\label{prop:randomness_limited_secrecy_capacity}
  The secrecy capacity of a wiretap channel $(\calX,W_{\rvY\rvZ|\rvX},\calY\times\calZ)$ with a rate-limited source of local randomness $(\calR,p_\rvR)$ and a less capable eavesdropper's channel\footnote{We used the less capable assumption to avoid dealing with the problem of channel prefixing. Days before submitting the current paper,~\cite{Watanabe2012} was posted on ArXiv and independently solved the general case. Proposition~\ref{prop:randomness_limited_secrecy_capacity} appears as~\cite[Corollary 12]{Watanabe2012}.} is
  \begin{align*}
    C_s = \max_{p_{\rvU\rvV\rvX\rvY\rvZ}\in\calP}\left(\avgI{\rvX;\rvY|\rvU}-\avgI{\rvX;\rvZ|\rvU}\right)
  \end{align*}
  where the set $\calP$ is the set of distributions $p_{\rvU\rvX\rvY\rvZ}$ that factorize as $p_{\rvU\rvX\rvY\rvZ} = p_{\rvU}p_{\rvV|\rvU}p_{\rvX|\rvV}W_{\rvY\rvZ|\rvX}$  and with $\avgI{\rvX;\rvZ|\rvU} \leq \avgH{\rvR}$.
\end{proposition}
\begin{proof}
  See Appendix~\ref{sec:conv-proof} and Appendix~\ref{sec:achievability}.
\end{proof}

\begin{remark}
Using standard techniques, one can show that the cardinality of $\calU$ is bounded by $\card{\calU}\leq 2$.
\end{remark} 
The expression in Proposition~\ref{prop:randomness_limited_secrecy_capacity} is similar to that obtained in~\cite[Corollary 2]{Csiszar1978}. The effect of the local source of randomness explicitly appears in the expression through the auxiliary time-sharing random variable $\rvU$ and the constraint $\avgI{\rvX;\rvZ|\rvU}\leq \avgH{\rvR}$. Proposition~\ref{prop:randomness_limited_secrecy_capacity} confirms the optimal structure of the encoder, which performs two distinct operations:
\begin{enumerate}
\item \emph{Uniformization:} the encoder generates nearly-uniform random numbers $\rvU_r$ at rate $\avgH{\rvR}$ from the local source of randomness;
\item \emph{Randomization:} the encoder uses a fraction $\avgI{\rvX;\rvZ|\rvU}$ of the randomness rate to randomize the choice of a codeword;
\end{enumerate}
The identification of the optimal encoder structure suggest that non-uniform randomization may affect the performance of a code, which we discuss in Section~\ref{sec:secure-comm-with}. Proposition~\ref{prop:randomness_limited_secrecy_capacity} also highlights that the common folklore in information-theoretic security, according to which secrecy is achievable provided the randomization can exhaust the capacity of the Eve's channel, is somewhat misleading. If the source provides a non-zero rate of randomness ($\avgH{\rvR}>0$), then the secrecy capacity with a rate-limited source of randomness is positive if and only if the secrecy capacity with unlimited randomness is positive. Intuitively, this happens because the channel seen by Eve is an ``effective channel'', which is partly controlled by Alice through time-sharing and the choice of the codebook.

Also note that if the rate of randomness vanishes, then no secure communication is possible. This confirms that, except for pathological channels (for instance, one for which $\avgI{\rvX;\rvZ}=0$ for any $\rvX$), one cannot replace the local source of randomness by a pseudo-random number generator without losing the information-theoretic secrecy guarantees.

\section{Rate-Limited Randomness: Practical Considerations}
\label{sec:discussion}

It is legitimate to wonder how the results of previous sections generalize to continuous channels and, in particular, to Gaussian channels. There are no conceptual difficulties in analyzing the randomization part of the encoder since $\avgI{\rvX;\rvZ|\rvU}$ remains finite with a power constraint; however, the simulation of Gaussian noise plays a key role in multi-user wiretap channels~\cite{Tekin2008} as a means to perform cooperative jamming.

\begin{figure}[t]
  \centering
    \scalebox{0.5}{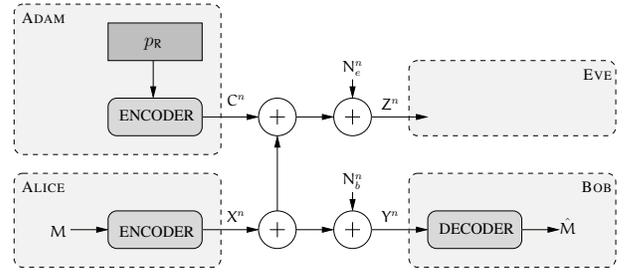}
  \caption{Cooperative jamming with rate-limited randomness.}
  \label{fig:coop_jam}
\end{figure}
We analyze the situation illustrated in Figure~\ref{fig:coop_jam}, in which an eavesdropper observes the output of an AWGN channel with noise variance $\sigma^2$ and suffers from the added interference of a cooperative jammer (Adam). The signal obtained by the eavesdropper is then
\begin{align*}
  \rvZ^n=\rvX^n+\rvC^n+\rvN_e^n,
\end{align*}
where $\rvX^n$ is the codeword transmitted by Alice, $\rvC^n$ is the interference introduced by Adam, and $\rvN_e^n$ is the channel noise. Cooperative jamming would consist in generating $\rvC^n$ i.i.d. according to a Gaussian distribution. With a local source of randomness, the following results holds.
\begin{proposition}
  With a local source of randomness $(\calR,p_{\rvR})$, a cooperative jammer can induce artificial Gaussian noise with power $\rho\leq \sigma^22^{2\avgH{\rvR}-1}$. 
\end{proposition}
\begin{proof}[Sketch of proof]
  The result follows by remarking that the objective of cooperative jamming is to increase the variance of Gaussian noise at the eavesdropper's terminal; therefore, the distribution of $\rvC^n+\rvN_e^n$ should be close to Gaussian, but $\rvC^n$ itself need not be Gaussian. In particular, the sequences $\rvC^n$ can be chosen from a codebook with average power constraint $\rho$; the result of channel resolvability over Gaussian channels~\cite{Han1994} guarantees there exists a codebook with rate arbitrarily close to $\frac{1}{2}\log(1+\frac{\rho}{\sigma^2})$ so that the distribution of $\rvC^n+\rvN_e^n$ is arbitrarily close to $\calN(\sigma^2+\rho)$. Since the rate of the codebook is given by the rate of the source $\avgH{\rvR}$, the result follows.
\end{proof}
Consequently, rate-limited randomness effectively translates into a power constraint on the Gaussian artificial noise that the encoder introduces to jam the eavesdropper. Therefore, rate-limited randomness reduces the effectiveness of cooperating jamming but does not preclude it.

\section{Non-Uniform Rate-Limited Randomness}
\label{sec:secure-comm-with}
The result of Proposition~\ref{prop:randomness_limited_secrecy_capacity} suggests that one should always ``uniformize'' the local source of randomness to create uniformly distributed random numbers. This operation, however, may be imperfect and one may wonder whether achieving secrecy is then still possible. A situation where the random numbers may not be perfectly uniform is if the local source of randomness is another message source; understanding this setting is crucial to assess whether secrecy constraints incur an overall rate loss or not~\cite{Bloch2011}.

For simplicity, we assume that the output of uniformization is a random variable $\rvU_r\in\intseq{1}{2^{nR_r}}$ with perhaps non-uniform distribution $p_{\rvU_r}$. In this case, we show that secrecy is still achievable, but at a lower rate limited by the R\'enyi entropy rate of order two $\frac{1}{n}R_2(\rvU_r)$ where
\begin{align*}
  R_2(\rvU_r) \eqdef - \log\left(\sum_{u\in\intseq{1}{2^{nR_r}}}p_{\rvU_r}(u)^2\right).
\end{align*}

\begin{proposition}
  \label{prop:non_uniform}
  A secrecy rate $R$ is achievable when randomization is performed with randomness $\rvU_r$ if it satisfies
  \begin{align*}
    R<\max_{p_{\rvU\rvX\rvY\rvZ}\in\calP}\left(\avgI{\rvX;\rvY|\rvU}-\avgI{\rvX;\rvZ|\rvU}\right),
  \end{align*}
  where $\calP$ is the set of distributions $\avgI{\rvX;\rvY|\rvU}$ that factorize as $p_{\rvU}p_{\rvX|\rvU}W_{\rvY\rvZ|\rvX}$ and such that $\avgI{\rvX;\rvZ|\rvU}<\frac{1}{n}R_2(\rvU_r)$.
\end{proposition}
\begin{proof}
  See Appendix~\ref{sec:proof-proposition-non-uniform}.
\end{proof}
It is not straightforward to establish a converse for Proposition~\ref{prop:non_uniform} because typical converse arguments make no assumption regarding the internal structure of the encoder. In particular, it seems difficult to include a constraint that would prevent any processing of $\rvU_r$. 

In general, $\frac{1}{n}R_2(\rvU_r)\leq \frac{1}{n}\avgH{\rvU_r}$, and the constraint in Proposition~\ref{prop:non_uniform} is therefore more stringent than in Proposition~\ref{prop:randomness_limited_secrecy_capacity}. The effect can be quite dramatic, and the following example shows that the gap between the rates in Proposition~\ref{prop:randomness_limited_secrecy_capacity} and Proposition~\ref{prop:non_uniform} can be large.
\begin{example}
  Assume the encoder performs randomization with a biased local source of randomness, which produces random numbers $\rvU_r\in\intseq{1}{2^{nR}}$ such that
  \begin{align*}
    \P{\rvU_r=1}=2^{-n\alpha R}\text{ and }\P{\rvU_r=i} = \frac{1-2^{-n\alpha R}}{2^{nR}-1}\text{ if }i\neq 1,
  \end{align*}
  where $\alpha\in]0;\frac{1}{2}[$ is a parameter that controls the uniformity of the distribution. Note that
  \begin{align*}
    \lim_{n\rightarrow\infty}\tfrac{1}{n}R_2(\rvU_r) = \alpha R\quad\text{whereas}\quad \lim_{n\rightarrow\infty}\tfrac{1}{n}\avgH{\rvU_r} = R.
  \end{align*}
  Consequently, without proper uniformization, the achievable rates predicted in Proposition~\ref{prop:non_uniform} could be arbitrarily small.
\end{example}

\section{Acknowledgements}
\label{sec:acknowledgements}
This work was supported in part by the U.S. National Science Foundation under grants CCF-0830666 and CCF-1017632. 

\appendices

\section{Converse Proof for Proposition~\ref{prop:randomness_limited_secrecy_capacity}}
\label{sec:conv-proof}
Let $\epsilon>0$ and let $R$ be an achievable rate. Then, there exists a $(2^{nR},n)$ code $\calC_n$ such that $P_e(\calC_n)\leq \epsilon$ and $L(\calC_n)\leq \epsilon$. Following the converse technique in~\cite{Csiszar1978}, we obtain 
\begin{align*}
  R\leq\tfrac{1}{n}\sum_{i=1}^n\left(\avgI{\rvM;\rvY_i|\rvY^{i-1}\tilde{\rvZ}^{i+1}}\!\!-\!\avgI{\rvM;\rvZ_i|\rvY^{i-1}\tilde{\rvZ}^{i+1}}\right)\!+\!\delta(\epsilon),
\end{align*}
where $\tilde{\rvY}^{i-1}\eqdef\{\rvY_j\}_{j=1}^{i-1}$, $\tilde{\rvZ}^{i+1}\eqdef\{\rvZ_j\}_{j=i+1}^n$ and $\delta(\epsilon)$ is a function of $\epsilon$ that goes to zero with $\epsilon$. Next, by definition of the encoder $e$ and by independence of $\rvR^n$ and $\rvM$,
\begin{align}
  \frac{1}{n}\avgH{\rvX^n|\rvM} =   \frac{1}{n}\avgH{e(\rvM,\rvR^n)|\rvM}\leq \frac{1}{n}\avgH{\rvR^n}=\avgH{\rvR}.\label{eq:bound_Raux_1}
\end{align}

Now, we also have
\begin{align}
  &\frac{1}{n}\avgH{\rvX^n|\rvM} \nonumber\\
  &=\frac{1}{n}\avgH{\rvX^n}-\frac{1}{n}\avgH{\rvM} + \frac{1}{n}\avgH{\rvM|\rvX^n}\nonumber\allowdisplaybreaks[0]\\
  &\geq \frac{1}{n}\avgH{\rvX^n} -\frac{1}{n}\avgH{\rvM} +\frac{1}{n}\avgH{\rvM|\rvX^n}+\frac{1}{n}\avgI{\rvM;\rvZ^n} - \delta(\epsilon)\nonumber\allowdisplaybreaks[0]\\
  & = \frac{1}{n}\avgH{\rvX^n} -\frac{1}{n}\avgH{\rvM|\rvZ^n} - \delta(\epsilon)+\frac{1}{n}\avgH{\rvM|\rvX^n}\nonumber\\
  &= \frac{1}{n}\avgH{\rvX^n} -\frac{1}{n}\avgH{\rvM\rvX^n|\rvZ^n} +\frac{1}{n}\avgH{\rvX^n|\rvM\rvZ^n}\nonumber\\
  &\phantom{----------------}+\frac{1}{n}\avgH{\rvM|\rvX^n}- \delta(\epsilon)\nonumber\\
  &=\frac{1}{n}\avgI{\rvX^n;\rvZ^n} + \frac{1}{n}\avgH{\rvX^n|\rvM\rvZ^n}-\delta(\epsilon)\nonumber\\
  &\geq\frac{1}{n}\avgI{\rvX^n;\rvZ^n}-\delta(\epsilon),\label{eq:bound_Raux_2}
\end{align}
where the last inequality follows because $\rvM\rightarrow\rvX^n\rightarrow\rvZ^n$ forms a Markov chain and $\avgH{\rvM|\rvX^n\rvZ^n}=\avgH{\rvM|\rvX^n}$. Then,
\begin{align}
 &\frac{1}{n}\avgI{\rvX^n;\rvZ^n}\nonumber\\
 &\phantom{--}= \frac{1}{n}\sum_{i=1}^n\left(\avgH{\rvZ_i|\tilde{\rvZ}^{i+1}}-\avgH{\rvZ_i|\rvX^n\tilde{\rvZ}^{i+1}}\right)\nonumber\\
 &\phantom{--}\geq \frac{1}{n}\sum_{i=1}^n\left(\avgH{\rvZ_i|\rvY^{i-1}\tilde{\rvZ}^{i+1}}-\avgH{\rvZ_i|\rvY^{i-1}\tilde{\rvZ}^{i+1}\rvX_i}\right)\nonumber\\
&\phantom{--}=\frac{1}{n}\sum_{i=1}^n\avgI{\rvX_i;\rvZ_i|\rvY^{i-1}\tilde{\rvZ}^{i+1}},\label{eq:bound_Raux_3}
\end{align}
where the inequality follows because conditioning does not increase entropy and $\tilde{\rvZ}^{i+1}\rvY^{i-1}\rightarrow\rvX_i\rightarrow\rvZ_i$ forms a Markov chain. Let us now define a random variable $\rvQ$ independent of all others and uniformly distributed on $\intseq{1}{n}$. For $i\in\intseq{1}{n}$, we also define $\rvU_i\eqdef \rvY^{i-1}\tilde\rvZ^{i+1}$ and $\rvV_i\eqdef \rvU_i\rvM$. Combining inequalities~\eqref{eq:bound_Raux_1},~\eqref{eq:bound_Raux_2}, and~\eqref{eq:bound_Raux_3}, and substituting the definition of $\rvQ$, $\rvU_i$, $\rvV_i$ above, we obtain 
\begin{align}
  R&\leq \avgI{\rvV_\rvQ;\rvY_\rvQ|\rvQ\rvU_\rvQ}-\avgI{\rvV_\rvQ;\rvZ_\rvQ|\rvQ\rvU_\rvQ}+\delta(\epsilon)\label{eq:single_lett_1}\\
  \avgH{\rvR} &\geq \avgI{\rvX_\rvQ;\rvZ_\rvQ|\rvQ\rvU_\rvQ}-\delta(\epsilon)\label{eq:single_lett_2}.
\end{align}
Finally, define $\rvU\eqdef \rvU_\rvQ\rvQ$, $\rvV\eqdef \rvV_\rvQ\rvQ$, $\rvX\eqdef \rvX_\rvQ$, $\rvY\eqdef \rvY_\rvQ$ and $\rvZ\eqdef \rvZ_\rvQ$. Note that $\rvU\rightarrow\rvV\rightarrow\rvX\rightarrow\rvY\rvZ$ forms a Markov chain and that the statistics $p_{\rvY\rvZ|\rvX}$ are those of the original channel $W_{\rvY\rvZ|\rvX}$. Substituting these definitions in~\eqref{eq:single_lett_1} and~\eqref{eq:single_lett_2}, we obtain
\begin{align*}
  R&\leq \avgI{\rvV;\rvY|\rvU}-\avgI{\rvV;\rvZ|\rvU}+\delta(\epsilon)\\
  \avgH{\rvR}&\geq \avgI{\rvX;\rvZ|\rvU} -\delta(\epsilon).
\end{align*}
Because the eavesdropper's channel is less capable, then $\avgI{\rvV;\rvY|\rvU}-\avgI{\rvV;\rvZ|\rvU} \leq \avgI{\rvX;\rvY|\rvU}-\avgI{\rvX;\rvZ|\rvU}$. Since $\epsilon$ can be chosen arbitrarily small, we obtained the desired converse.

\section{Achievability Proof for Proposition~\ref{prop:randomness_limited_secrecy_capacity}}
\label{sec:achievability}
The proof relies on binning, superposition coding, and stochastic encoding as in~\cite[Lemma 2]{Csiszar1978}; however, since the local source of randomness is explicit and since we impose a strong secrecy criterion, some details must be laid out carefully. We denote the set of $\epsilon$-strongly typical sequences with respect to $p_\rvX$ by  $\sts{\epsilon}{n}{\rvX}$ and the set of conditional $\epsilon$-strongly typical sequence with respect to $p_{\rvY\rvX}$ and $x^n\in\sts{\epsilon}{n}{\rvX}$ by $\sts{\epsilon}{n}{\rvY|x^n}$. 

We first show the existence of a code $\calC_n$ assuming an unlimited amount of uniform randomness is available. We fix a joint distribution $p_{\rvU\rvX}$ on $\calU\times\calX$ such that\footnote{If such a probability distribution does not exist, then the result of Proposition~\ref{prop:randomness_limited_secrecy_capacity} is trivial and there is nothing to prove.} $\avgI{\rvX;\rvZ|\rvU}\leq\avgH{\rvR}$ and $\avgI{\rvX;\rvY|\rvU}-\avgI{\rvX;\rvZ|\rvU}>0$,  and we construct a code $\calC_n$ for the broadcast channel with confidential messages $(\calX,p_{\rvY\rvZ|\rvX},\calY\times\calZ)$. Let $\epsilon>0$, $R>0$, $R_r>0$, $R_0>0$ and $n\in\mathbb{N}$. We randomly construct a code as follows. We generate $2^{nR_0}$ sequences independently at random according to $p_\rvU$, which we label $u^n(i)$ for $i\in\intseq{1}{2^{nR_0}}$. For each sequence $u^n(i)$, we generate $2^{n(R+R_r)}$ sequences independently a random according to $p_{\rvX|\rvU}$, which we label $x^{n}(i,j,k)$ with $j\in\intseq{1}{2^{nR}}$ and $k\in\intseq{1}{2^{nR_r}}$. To transmit a message $i\in\intseq{1}{2^{nR_0}}$ and $j\in\intseq{1}{2^{nR}}$, the transmitter obtains a realization $k$ of a uniform random number $\rvU_{r}\in\intseq{1}{2^{nR_r}}$, and transmits $x^n(i,j,k)$ over the channel. Upon receiving $y^n$, Bob decodes $i$ as the received index if it is the unique one such that $(u^n(i),y^n)\in\sts{\epsilon}{n}{UY}$; otherwise he declares an error. Bob then decode $(j,k)$ as the other pair of indices if it is the unique one such that $(x^n(i,j,k),y^n)\in\sts{\epsilon}{n}{UXY}$. Similarly, upon receiving $z^n$, Eve decodes $i$ as the received index if it is the unique one such that $(u^n(i),z^n)\in\sts{\epsilon}{n}{UZ}$; otherwise she declares an error.
\begin{lemma}
  \label{lm:reliability}
  If $R_0<\min(\avgI{\rvU;\rvY},\avgI{\rvU;\rvZ})$ and $R+R_r<\avgI{\rvX;\rvY|\rvU}$, then $\E{P_e(\rvC_n)}\leq 2^{-\alpha n}$ for some $\alpha>0$.
\end{lemma}
\begin{proof}
  The proof follows from a standard random coding argument and is omitted.
\end{proof}
\begin{lemma}
  \label{lm:secrecy}
  If $R_r>\avgI{\rvX;\rvZ|\rvU}$, then we have $\E[\rvC_n]{\V{p_{\rvM\rvZ^n},p_{\rvM}p_{\rvZ^n}}}\leq 2^{-\beta n}$ for some $\beta >0$, where $\mathbb{V}$ denotes the variational distance.
\end{lemma}
\begin{proof}
  Lemma~\ref{lm:secrecy} is a special case of Lemma~\ref{lm:resolvability_renyi} proved in Appendix~\ref{sec:proof-proposition-non-uniform}.
\end{proof}
Using Markov's inequality, we conclude that there exists at least one code $\calC_n$ satisfying the rate inequalities in Lemma~\ref{lm:reliability} and Lemma~\ref{lm:secrecy}, such that $P_e(\calC_n)\leq 3\cdot 2^{-\alpha n}$ and $\V{p_{\rvM\rvM_0\rvZ^n},p_{\rvM}p_{\rvM_0\rvZ^n}}\leq 3\cdot 2^{-\beta n}$. Finally, the uniform numbers $\rvU_{r}$ can be approximately obtained from $(\calR,p_\rvR)$ with an appropriate function $\phi$.
\begin{lemma}[adapted from~\cite{Ahlswede1998}]
  If $R_r<\avgH{\rvR}$, then there exists $\phi$ such that $\V{p_{\phi(\rvR^n)},p_{\rvU_r}}\leq 2^{-n\eta}$ for some $\eta>0$.
\end{lemma}
Consequently, it is not hard to show that, even if the code $\calC_n$ is used with $\phi(\rvR^n)$ in place of $\rvU_{r}$, then 
\begin{align*}
  P_e(\calC_n) \leq 2^{-\kappa n}\quad\text{and}\quad \V{p_{\rvM\rvM_0\rvZ^n},p_{\rvM}p_{\rvM_0\rvZ^n}}\leq 2^{-\kappa n}.
\end{align*}
for some $\kappa>0$. The fact that $L(\calC_n)\leq 2^{-\kappa' n}$ for some $\kappa' >0$ follows from~\cite[Lemma 1]{Csiszar1996}. Combining all rate constraints in the previous lemmas, and since $\epsilon$ can be chosen arbitrarily small, we see that any rate $R<\avgI{\rvX;\rvY|\rvU}-\avgI{\rvX;\rvZ|\rvU}$ such that $\avgI{\rvX;\rvZ|\rvU} \leq \avgH{\rvR}$ is achievable. Note that the constraint on $R_0$ plays no role since it represents a negligible rate of time sharing information to synchronize transmitter and receiver.

\section{Proof of Proposition~\ref{prop:non_uniform}}
\label{sec:proof-proposition-non-uniform}
The proof is similar to that Appendix~\ref{sec:achievability}, with Lemma~\ref{lm:resolvability_renyi} in place of Lemma~\ref{lm:secrecy}. Lemma~\ref{lm:secrecy} is obtained in the special case of $\rvU_{r}$ uniform.
\begin{lemma}
  \label{lm:resolvability_renyi}
    If $\frac{1}{n}R_2(\rvU_{r})>\avgI{\rvX;\rvZ|\rvU}$, then we have $\E[\rvC_n]{\V{p_{\rvM\rvZ^n},p_{\rvM}p_{\rvZ^n}}}\leq 2^{-\beta n}$ for some $\beta >0$.
\end{lemma}
The proof relies on a careful analysis and modification of the ``cloud-mixing'' lemma~\cite{ThesisCuff2009} and the notation is that of Appendix~\ref{sec:achievability}. We define the distribution $q_{\rvU^n\rvX^n\rvZ^n}$ on $\calU^n\times\calX^n\times\calZ^n$ as
  \begin{align*}
    q_{\rvU^n\rvX^n\rvZ^n}(u^n,x^n,z^n) = W_{\rvZ^n|\rvX^n}(z^n|x^n)p_{\rvX^n\rvU^n}(x^n,u^n).
  \end{align*}
First note that the variational distance $\V{p_{\rvM\rvZ^n},p_{\rvM}p_{\rvZ^n}}$ can be bounded as follows.
  \begin{align*}
    &\V{p_{\rvM\rvZ^n},p_{\rvM}p_{\rvZ^n}}\\
    &\quad\leq \V{p_{\rvM\rvU^n\rvZ^n},p_{\rvM}p_{\rvU^n\rvZ^n}}\\
    &\quad =\E[\rvU^n\rvM]{\V{p_{\rvZ^n|\rvM\rvU^n},p_{\rvZ^n|\rvU^n}}}\\
    &\quad \leq \E[\rvU^n\rvM]{\V{p_{\rvZ^n|\rvM\rvU^n},q_{\rvZ^n|\rvU^n}}+\V{q_{\rvZ^n|\rvU^n},p_{\rvZ^n|\rvU^n}}}\\
    &\quad \leq 2 \E[\rvU^n\rvM]{\V{p_{\rvZ^n|\rvM\rvU^n},q_{\rvZ^n|\rvU^n}}}
  \end{align*}
Then, let $\rvU_1^n$ be the sequence in $\calU^n$ corresponding to $\rvM_0=1$. By symmetry of the random code construction, the average of the variational distance $\V{p_{\rvM\rvZ^n},p_{\rvM}p_{\rvZ^n}}$ over randomly generated codes $\rvC_n$ satisfies
  \begin{multline*}
    \E[\rvC_n]{\V{p_{\rvM\rvZ^n},p_{\rvM}p_{\rvZ^n}}}\\
    \leq 2\E[\rvC_n]{\V{p_{\rvZ^n|\rvU^n=\rvU_1^n\rvM=1},q_{\rvZ^n|\rvU^n=\rvU_1^n}}},
  \end{multline*}
  where
  \begin{align*}
    p_{\rvZ^n|\rvU^n=\rvU_1^n\rvM=1}(z^n)=\sum_{k=1}^{2^{nR_r}}W_{\rvZ^n|\rvX^n}(z^n|x^n(1,1,k))p_{\rvU_{r}}(k).
  \end{align*}
The average over the random codes can be split between the average of $\rvU_1^n$ and the random code $\rvC_n(u_1^n)$ for a fixed value of $u_1^n$, so that
  \begin{align*}
    &\E[\rvC_n]{\V{p_{\rvZ^n|\rvU^n=\rvU_1^n\rvM=1},q_{\rvZ^n|\rvU^n=\rvU_1^n}}} \\
    &= \sum_{u_1^n\in\calU^n}p_{\rvU^n}(u_1^n)\E[\rvC_n(u_1^n)]{\V{p_{\rvZ^n|\rvU^n=u_1^n\rvM=1},q_{\rvZ^n|\rvU^n=u_1^n}}}\\
    &\leq 2\P{\rvU^n\notin \sts{\epsilon}{n}{\rvU}}\\
    &+ \!\!\!\!\sum_{u_1^n\in\sts{\epsilon}{n}{\rvU}}p_{\rvU^n}(u_1^n)\E[\rvC_n(u_1^n)]{\V{p_{\rvZ^n|\rvU^n=u_1^n\rvM=1},q_{\rvZ^n|\rvU^n=u_1^n}}},
  \end{align*}
  where the last inequality follows from the fact that the variational distance is always less than 2. By construction, the first term on the right-hand side vanishes as $n$ gets large; we now proceed to bound the expectation in the second term. First note that, for any $z^n\in\calZ^n$,
  \begin{align*}
&\E[\rvC_n(u_1^n)]{p_{\rvZ^n|\rvU^n=u_1^n\rvM=1}(z^n)}\\
&=\E[\rvC_n(u_1^n)]{\sum_{k=1}^{2^{nR_r}}W_{\rvZ^n|\rvX^n}(z^n|x^n(1,1,k))p_{\rvU^r}(k)}\\
&=\sum_{k=1}^{2^{nR_r}}\E[\rvC_n(u_1^n)]{ W_{\rvZ^n|\rvX^n}(z^n|x^n(1,1,k))}p_{\rvU^r}(k)\\
&=q_{\rvZ^n|\rvU^n=u_1^n}(z^n).
  \end{align*}
  We now let $\mathbf{1}$ denote the indicator function and we define
  \begin{align*}
    p^{(1)}(z^n)&\eqdef\sum_{k=1}^{2^{nR_r}}W_{\rvZ^n|\rvX^n}(z^n|x^n(1,1,k))p_{\rvU_{r}}(k)\\
    &\phantom{-------}\mathbf{1}\{(x^n(1,1,k),z^n)\in\sts{\epsilon}{n}{\rvX\rvZ|u_1^n}\},\\
    p^{(2)}(z^n)&\eqdef\sum_{k=1}^{2^{nR_r}}W_{\rvZ^n|\rvX^n}(z^n|x^n(1,1,k))p_{\rvU_{r}}(k)\\
    &\phantom{-------}\mathbf{1}\{(x^n(1,1,k),z^n)\notin\sts{\epsilon}{n}{\rvX\rvZ|u_1^n}\},
  \end{align*}
 so that we can upper bound $\V{p_{\rvZ^n|\rvU^n=u_1^n\rvM=1},q_{\rvZ^n|\rvU^n=u_1^n}}$ as
 \begin{align}
   &\V{p_{\rvZ^n|\rvU^n=u_1^n\rvM=1},q_{\rvZ^n|\rvU^n=u_1^n}}\nonumber{}\\
   &\leq \sum_{z^n\notin\sts{\epsilon}{n}{\rvZ|u_1^n}}\abs{p_{\rvZ^n|\rvU^n=u_1^n\rvM=1}(z^n)-q_{\rvZ^n|\rvU^n=u_1^n}(z^n)}\label{eq:1}\\
   &\quad +\sum_{z^n\in\sts{\epsilon}{n}{\rvZ|u_1^n}}\abs{p^{(1)}(z^n)-\E{p^{(1)}(z^n)}}\label{eq:2}\\
   &\quad +\sum_{z^n\in\sts{\epsilon}{n}{\rvZ|u_1^n}}\abs{p^{(2)}(z^n)-\E{p^{(2)}(z^n)}}\label{eq:3}.
 \end{align}
 Taking the expectation of the term in~\eqref{eq:1} over $\rvC_n(u_1^n)$, we obtain
 \begin{align*}
   &\E{\sum_{z^n\notin\sts{\epsilon}{n}{\rvZ|u_1^n}}\abs{p_{\rvZ^n|\rvU^n=u_1^n\rvM=1}(z^n)-q_{\rvZ^n|\rvU^n=u_1^n}(z^n)}}\\
   &\leq \sum_{z^n\notin\sts{\epsilon}{n}{\rvZ|u_1^n}}\E{\max(p_{\rvZ^n|\rvU^n=u_1^n\rvM=1}(z^n),q_{\rvZ^n|\rvU^n=u_1^n}(z^n))}\\
   & = \sum_{z^n\notin\sts{\epsilon}{n}{\rvZ|u_1^n}}q_{\rvZ^n|\rvU^n=u_1^n}(z^n),
 \end{align*}
which vanishes as $n$ goes to infinity. Similarly, taking the expectation of the term in~\eqref{eq:3} over $\rvC_n(u_1^n)$, we obtain
\begin{align*}
  &\E{\sum_{z^n\in\sts{\epsilon}{n}{\rvZ|u_1^n}}\abs{p^{(2)}(z^n)-\E{p^{(2)}(z^n)}}}\\
  &\leq\E{\sum_{z^n\in\calZ^n}\abs{p^{(2)}(z^n)-\E{p^{(2)}(z^n)}}}\\
  &\leq \sum_{z^n\in\calZ^n}\E{p^{(2)}(z^n)}\\
  &=\sum_{z^n\in\calZ^n}\mathbb{E}\left(W_{\rvZ^n|\rvX^n}(z^n|\rvX^n(1,1,1))\right.\\
  &\phantom{----------}\left.\mathbf{1}\{(\rvX^n(1,1,1),z^n)\notin\sts{\epsilon}{n}{\rvX\rvZ|u_1^n}\right)\\
  &=\sum_{(x^n,z^n)\notin\sts{\epsilon}{n}{\rvX\rvZ|u_1^n}}q_{\rvZ^n\rvX^n|\rvU^n=u_1^n}(z^n,x^n),
\end{align*}
which vanishes as $n$ goes to infinity. Finally, we focus on the expectation of the term in~\eqref{eq:2} over $\rvC_n(u_1^n)$. For $z^n\in \sts{\epsilon}{n}{\rvZ|u_1^n}$, Jensen's inequality and the concavity of $x\mapsto \sqrt{x}$ guarantee that
 \begin{align*}
   \E{\abs{p^{(1)}(z^n)-\E{p^{(1)}(z^n)}}} \leq \sqrt{\Var{p^{(1)}(z^n)}}.
 \end{align*}
 In addition,
 \begin{align*}
   &\Var{p^{(1)}(z^n)}=
   \sum_{k=1}^{2^{nR_r}}p_{\rvU_{r}}(k)^2\text{Var}\left({W_{\rvZ^n|\rvX^n}(z^n|\rvX^n(1,1,k))}\right.\\
   &\phantom{-----------}\left. \mathbf{1}\{(\rvX^n(1,1,k),z^n)\in\sts{\epsilon}{n}{\rvX\rvZ|u_1^n}\}\right)
 \end{align*}
 Note that
 \begin{align*}
   &\text{Var}\left({W_{\rvZ^n|\rvX^n}(z^n|\rvX^n(1,1,k))}\mathbf{1}\{(\rvX^n(1,1,k),z^n)\in\sts{\epsilon}{n}{\rvX\rvZ|u_1^n}\}\right)\\
   &=\sum_{x^n\in\calX^n}p_{\rvX^n|\rvU^n=u_1^n}(x^n)\\
   &\phantom{------}
   \left({W_{\rvZ^n|\rvX^n}(z^n|x^n)}\mathbf{1}\{(x^n,z^n)\in\sts{\epsilon}{n}{\rvX\rvZ|u_1^n}\}\right)^2\\
   &=\sum_{x^n:(x^n,z^n)\in\sts{\epsilon}{n}{\rvX\rvZ|u_1^n}}p_{\rvX^n|\rvU^n=u_1^n}(x^n) W_{\rvZ^n|\rvX^n}(z^n|x^n)^2\\
   &\stackrel{(a)}{\leq}2^{-n(\avgH{\rvZ|\rvX}-\delta(\epsilon))}\\
   &\phantom{---}\sum_{x^n:(x^n,z^n)\in\sts{\epsilon}{n}{\rvX\rvZ|u_1^n}}p_{\rvX^n|\rvU^n=u_1^n}(x^n) W_{\rvZ^n|\rvX^n}(z^n|x^n)\\
   &\leq 2^{-n(\avgH{\rvZ|\rvX}-\delta(\epsilon))} q_{\rvZ^n|\rvU^n=u_1^n}(z^n)\\
   &\stackrel{(b)}{\leq}2^{-n(\avgH{\rvZ|\rvX}+\avgH{\rvZ|\rvU}-\delta(\epsilon))},
 \end{align*}
where $(a)$ and $(b)$ follow from the AEP; therefore, 
 \begin{align*}   
   \Var{p^{(1)}(z^n)}&\leq 2^{-n(\avgH{\rvZ|\rvX}+\avgH{\rvZ|\rvU}-\delta(\epsilon))}\sum_{k=1}^{2^{nR_r}}p_{\rvU_{r}}(k)^2\\
   &\leq 2^{-n(\avgH{\rvZ|\rvX}+\avgH{\rvZ|\rvU}-\delta(\epsilon))+\frac{R_2(\rvU_{r})}{n}}.
 \end{align*}
 and
 \begin{multline*}
   \sum_{z^n\in\sts{\epsilon}{n}{\rvZ|u_1^n}}\E{\abs{p^{(1)}(z^n)-\E{p^{(1)}(z^n)}}}\\
   \begin{split}
     &\leq 2^{n\avgH{\rvZ|\rvU}}2^{-\frac{n}{2}(\avgH{\rvZ|\rvX}+\avgH{\rvZ|\rvU}-\delta(\epsilon)+\frac{R_2(\rvU_{r})}{n})}\\
     &= 2^{-\frac{n}{2}(\frac{R_2(\rvU_{r})}{n}-\avgI{\rvX;\rvZ|\rvU}-\delta(\epsilon))}
   \end{split}
 \end{multline*}
 Hence, if $\frac{R_2(\rvU_{r})}{n}>\avgI{\rvX;\rvZ|\rvU}+\delta(\epsilon)$, the sum vanishes as $n$ goes to infinity, which concludes the proof. Note that if $\rvU_r$ is uniform, then $R_2(\rvU_r)=nR_r$, and we obtain Lemma~\ref{lm:secrecy}.

\newpage
\bibliographystyle{IEEEtran}
\bibliography{isit2012}

\end{document}

%% file: rv_defs.tex
\DeclareMathAlphabet{\eurm}{U}{eur}{m}{n}
\DeclareMathAlphabet{\mathbsf}{OT1}{cmss}{bx}{n}
\DeclareMathAlphabet{\mathssf}{OT1}{cmss}{m}{sl}
\DeclareMathAlphabet{\mathcsf}{OT1}{cmss}{sbc}{n}

\newcommand{\randomvalue}[1]{\eurm{\uppercase{#1}}}

\DeclareSymbolFont{bsfletters}{OT1}{cmss}{bx}{n}  
\DeclareSymbolFont{ssfletters}{OT1}{cmss}{m}{n}
\DeclareMathSymbol{\bsfGamma}{0}{bsfletters}{'000}
\DeclareMathSymbol{\ssfGamma}{0}{ssfletters}{'000}
\DeclareMathSymbol{\bsfDelta}{0}{bsfletters}{'001}
\DeclareMathSymbol{\ssfDelta}{0}{ssfletters}{'001}
\DeclareMathSymbol{\bsfTheta}{0}{bsfletters}{'002}
\DeclareMathSymbol{\ssfTheta}{0}{ssfletters}{'002}
\DeclareMathSymbol{\bsfLambda}{0}{bsfletters}{'003}
\DeclareMathSymbol{\ssfLambda}{0}{ssfletters}{'003}
\DeclareMathSymbol{\bsfXi}{0}{bsfletters}{'004}
\DeclareMathSymbol{\ssfXi}{0}{ssfletters}{'004}
\DeclareMathSymbol{\bsfPi}{0}{bsfletters}{'005}
\DeclareMathSymbol{\ssfPi}{0}{ssfletters}{'005}
\DeclareMathSymbol{\bsfSigma}{0}{bsfletters}{'006}
\DeclareMathSymbol{\ssfSigma}{0}{ssfletters}{'006}
\DeclareMathSymbol{\bsfUpsilon}{0}{bsfletters}{'007}
\DeclareMathSymbol{\ssfUpsilon}{0}{ssfletters}{'007}
\DeclareMathSymbol{\bsfPhi}{0}{bsfletters}{'010}
\DeclareMathSymbol{\ssfPhi}{0}{ssfletters}{'010}
\DeclareMathSymbol{\bsfPsi}{0}{bsfletters}{'011}
\DeclareMathSymbol{\ssfPsi}{0}{ssfletters}{'011}
\DeclareMathSymbol{\bsfOmega}{0}{bsfletters}{'012}
\DeclareMathSymbol{\ssfOmega}{0}{ssfletters}{'012}

\newcommand{\rvC}{{\randomvalue{C}}}

\newcommand{\rvM}{{\randomvalue{M}}}	
\newcommand{\rvN}{{\randomvalue{N}}}	
\newcommand{\rvQ}{{\randomvalue{Q}}}	
\newcommand{\rvR}{{\randomvalue{R}}}	
\newcommand{\rvU}{{\randomvalue{U}}}	
\newcommand{\rvV}{{\randomvalue{V}}}	
\newcommand{\rvX}{{\randomvalue{X}}}  	
\newcommand{\rvY}{{\randomvalue{Y}}}	
\newcommand{\rvZ}{{\randomvalue{Z}}}	

\newcommand{\calC}{{\mathcal{C}}}

\newcommand{\calM}{{\mathcal{M}}}
\newcommand{\calN}{{\mathcal{N}}}

\newcommand{\calP}{{\mathcal{P}}}

\newcommand{\calR}{{\mathcal{R}}}

\newcommand{\calU}{{\mathcal{U}}}

\newcommand{\calX}{{\mathcal{X}}}
\newcommand{\calY}{{\mathcal{Y}}}

\newcommand{\calZ}{{\mathcal{Z}}}

%% file: fct_defs.tex
\newcommand{\E}[2][]{{\mathbb{E}_{#1}}{\left(#2\right)}}       
\newcommand{\Var}[1]{{\text{\textnormal{Var}}{\left(#1\right)}}}       
\renewcommand{\P}[2][]{{\mathbb{P}_{#1}}{\left(#2\right)}}

\newcommand{\V}[1]{{{\mathbb{V}}\!\left(#1\right)}}
\newcommand{\avgI}[1]{{{\mathbb{I}}\!\left(#1\right)}}

\newcommand{\avgH}[1]{{\mathbb{H}}\!\left(#1\right)}

\newcommand{\sts}[3]{T_{#1}^{#2}\!(#3)}

\newcommand{\card}[1]{\ensuremath{\left|{#1}\right|}}           
\newcommand{\abs}[1]{\ensuremath{\left|#1\right|}}              
\newcommand{\eqdef}{\ensuremath{\triangleq}}                    
\newcommand{\intseq}[2]{\ensuremath{\llbracket{#1},{#2}\rrbracket}}  



%% file: th_defs.tex
\newtheorem{lemma}{Lemma}
\newtheorem{example}{Example}

\newtheorem{proposition}{Proposition}
\newtheorem{definition}{Definition}
\newtheorem{remark}{Remark}

%% file: graphics/randomness_limited_wiretap.pdf_t
\begin{picture}(0,0)%
\includegraphics{randomness_limited_wiretap.pdf}%
\end{picture}%
\setlength{\unitlength}{4144sp}%
\begingroup\makeatletter\ifx\SetFigFont\undefined%
\gdef\SetFigFont#1#2#3#4#5{%
  \reset@font\fontsize{#1}{#2pt}%
  \fontfamily{#3}\fontseries{#4}\fontshape{#5}%
  \selectfont}%
\fi\endgroup%
\begin{picture}(7224,2274)(1789,-3673)
\put(3511,-2806){\makebox(0,0)[b]{\smash{{\SetFigFont{12}{14.4}{\rmdefault}{\mddefault}{\updefault}{\color[rgb]{0,0,0}ENCODER}%
}}}}
\put(4456,-2671){\makebox(0,0)[b]{\smash{{\SetFigFont{12}{14.4}{\rmdefault}{\mddefault}{\updefault}{\color[rgb]{0,0,0}$\rvX^n$}%
}}}}
\put(3466,-1906){\makebox(0,0)[b]{\smash{{\SetFigFont{14}{16.8}{\rmdefault}{\mddefault}{\updefault}{\color[rgb]{0,0,0}$p_{\rvR}$}%
}}}}
\put(1891,-1636){\makebox(0,0)[lb]{\smash{{\SetFigFont{12}{14.4}{\rmdefault}{\mddefault}{\updefault}{\color[rgb]{0,0,0}\textsc{Alice}}%
}}}}
\put(5401,-2806){\makebox(0,0)[b]{\smash{{\SetFigFont{14}{16.8}{\rmdefault}{\mddefault}{\updefault}{\color[rgb]{0,0,0}$W_{\rvY\rvZ|\rvX}$}%
}}}}
\put(2431,-2806){\makebox(0,0)[rb]{\smash{{\SetFigFont{12}{14.4}{\rmdefault}{\mddefault}{\updefault}{\color[rgb]{0,0,0}$\rvM$}%
}}}}
\put(8911,-1636){\makebox(0,0)[rb]{\smash{{\SetFigFont{12}{14.4}{\rmdefault}{\mddefault}{\updefault}{\color[rgb]{0,0,0}\textsc{Bob}}%
}}}}
\put(8911,-2986){\makebox(0,0)[rb]{\smash{{\SetFigFont{12}{14.4}{\rmdefault}{\mddefault}{\updefault}{\color[rgb]{0,0,0}\textsc{Eve}}%
}}}}
\put(7336,-2131){\makebox(0,0)[b]{\smash{{\SetFigFont{12}{14.4}{\rmdefault}{\mddefault}{\updefault}{\color[rgb]{0,0,0}DECODER}%
}}}}
\put(6301,-1996){\makebox(0,0)[b]{\smash{{\SetFigFont{12}{14.4}{\rmdefault}{\mddefault}{\updefault}{\color[rgb]{0,0,0}$\rvY^n$}%
}}}}
\put(6301,-3346){\makebox(0,0)[b]{\smash{{\SetFigFont{12}{14.4}{\rmdefault}{\mddefault}{\updefault}{\color[rgb]{0,0,0}$\rvZ^n$}%
}}}}
\put(8371,-2086){\makebox(0,0)[lb]{\smash{{\SetFigFont{12}{14.4}{\rmdefault}{\mddefault}{\updefault}{\color[rgb]{0,0,0}$\hat{\rvM}$}%
}}}}
\end{picture}%

%% file: graphics/coop_jam.pdf_t
\begin{picture}(0,0)%
\includegraphics{coop_jam.pdf}%
\end{picture}%
\setlength{\unitlength}{4144sp}%
\begingroup\makeatletter\ifx\SetFigFont\undefined%
\gdef\SetFigFont#1#2#3#4#5{%
  \reset@font\fontsize{#1}{#2pt}%
  \fontfamily{#3}\fontseries{#4}\fontshape{#5}%
  \selectfont}%
\fi\endgroup%
\begin{picture}(7224,3174)(1789,-4573)
\put(3511,-2806){\makebox(0,0)[b]{\smash{{\SetFigFont{12}{14.4}{\rmdefault}{\mddefault}{\updefault}{\color[rgb]{0,0,0}ENCODER}%
}}}}
\put(3466,-1906){\makebox(0,0)[b]{\smash{{\SetFigFont{14}{16.8}{\rmdefault}{\mddefault}{\updefault}{\color[rgb]{0,0,0}$p_{\rvR}$}%
}}}}
\put(1891,-1636){\makebox(0,0)[lb]{\smash{{\SetFigFont{12}{14.4}{\rmdefault}{\mddefault}{\updefault}{\color[rgb]{0,0,0}\textsc{Adam}}%
}}}}
\put(7336,-4156){\makebox(0,0)[b]{\smash{{\SetFigFont{12}{14.4}{\rmdefault}{\mddefault}{\updefault}{\color[rgb]{0,0,0}DECODER}%
}}}}
\put(6301,-4021){\makebox(0,0)[b]{\smash{{\SetFigFont{12}{14.4}{\rmdefault}{\mddefault}{\updefault}{\color[rgb]{0,0,0}$\rvY^n$}%
}}}}
\put(6301,-2671){\makebox(0,0)[b]{\smash{{\SetFigFont{12}{14.4}{\rmdefault}{\mddefault}{\updefault}{\color[rgb]{0,0,0}$\rvZ^n$}%
}}}}
\put(8956,-2311){\makebox(0,0)[rb]{\smash{{\SetFigFont{12}{14.4}{\rmdefault}{\mddefault}{\updefault}{\color[rgb]{0,0,0}\textsc{Eve}}%
}}}}
\put(8956,-3661){\makebox(0,0)[rb]{\smash{{\SetFigFont{12}{14.4}{\rmdefault}{\mddefault}{\updefault}{\color[rgb]{0,0,0}\textsc{Bob}}%
}}}}
\put(8326,-4156){\makebox(0,0)[lb]{\smash{{\SetFigFont{12}{14.4}{\rmdefault}{\mddefault}{\updefault}{\color[rgb]{0,0,0}$\hat{\rvM}$}%
}}}}
\put(2431,-4201){\makebox(0,0)[rb]{\smash{{\SetFigFont{12}{14.4}{\rmdefault}{\mddefault}{\updefault}{\color[rgb]{0,0,0}$\rvM$}%
}}}}
\put(3511,-4201){\makebox(0,0)[b]{\smash{{\SetFigFont{12}{14.4}{\rmdefault}{\mddefault}{\updefault}{\color[rgb]{0,0,0}ENCODER}%
}}}}
\put(1891,-3661){\makebox(0,0)[lb]{\smash{{\SetFigFont{12}{14.4}{\rmdefault}{\mddefault}{\updefault}{\color[rgb]{0,0,0}\textsc{Alice}}%
}}}}
\put(5851,-2221){\makebox(0,0)[b]{\smash{{\SetFigFont{12}{14.4}{\rmdefault}{\mddefault}{\updefault}{\color[rgb]{0,0,0}$\rvN_e^n$}%
}}}}
\put(5851,-3571){\makebox(0,0)[b]{\smash{{\SetFigFont{12}{14.4}{\rmdefault}{\mddefault}{\updefault}{\color[rgb]{0,0,0}$\rvN_b^n$}%
}}}}
\put(4456,-4021){\makebox(0,0)[b]{\smash{{\SetFigFont{12}{14.4}{\rmdefault}{\mddefault}{\updefault}{\color[rgb]{0,0,0}$\rvX^n$}%
}}}}
\put(4456,-2626){\makebox(0,0)[b]{\smash{{\SetFigFont{12}{14.4}{\rmdefault}{\mddefault}{\updefault}{\color[rgb]{0,0,0}$\rvC^n$}%
}}}}
\end{picture}%